\pdfoutput=1

\RequirePackage[bookmarksnumbered,unicode]{hyperref}

\documentclass[sigconf,screen,nonacm]{acmart}

\startPage{1}
\usepackage{main}

\usepackage{mathtools}
\usepackage{mathpartir}
\usepackage{tikz}
\usepackage{tikzit}
\usepackage{accents}
\usepackage{catchfilebetweentags}

\usepackage[only,llbracket,rrbracket]{stmaryrd}

\usepackage{caption}
\usepackage{subcaption}
\usepackage[clock]{ifsym}
\usepackage{stackengine}
\usepackage{widows-and-orphans}

\AtBeginDocument{
  \theoremstyle{acmplain}
  
  \theoremstyle{acmdefinition}
  
  \newtheorem{mainidea}{Main Idea}
}
\AtBeginEnvironment{example}{%
  \pushQED{\qed}%
}
\AtEndEnvironment{example}{\popQED\endexample}
\AtBeginEnvironment{mainidea}{%
  \pushQED{\qed}%
}
\AtEndEnvironment{mainidea}{\popQED\endmainidea}

\Crefname{axiom}{Axiom}{Axioms}
\Crefname{mainidea}{Main Idea}{Main Ideas}

\title{Decalf: A Directed, Effectful Cost-Aware Logical Framework}

\author{Harrison Grodin}
\orcid{0000-0002-0947-3520}

\affiliation{
  \institution{Carnegie Mellon University}
  \department{Computer Science Department}
  \streetaddress{5000 Forbes Ave.}
  \city{Pittsburgh}
  \state{PA}
  \postcode{15213}
  \country{USA}
}
\email{hgrodin@cs.cmu.edu}

\author{Yue Niu}
\orcid{0000-0003-4888-6042}

\affiliation{
  \institution{National Institute of Informatics}
  \city{Tokyo}
  \postcode{15213}
  \country{Japan}
}
\email{yue_niu@nii.ac.jp}

\author{Jonathan Sterling}
\orcid{0000-0002-0585-5564}

\affiliation{
  \institution{University of Cambridge}
  \department{Computer Laboratory}
  \streetaddress{15 JJ Thomson Avenue}
  \city{Cambridge}
  \postcode{CB3 0FD}
  \country{UK}
}
\email{js2878@cl.cam.ac.uk}

\author{Robert Harper}
\orcid{0000-0002-9400-2941}

\affiliation{
  \institution{Carnegie Mellon University}
  \department{Computer Science Department}
  \streetaddress{5000 Forbes Ave.}
  \city{Pittsburgh}
  \state{PA}
  \postcode{15213}
  \country{USA}
}
\email{rwh@cs.cmu.edu}

\authorsaddresses{}

\begin{document}

\begin{abstract}
  We present \decalf{}, a \emph{d}irected, \emph{e}ffectful \emph{c}ost-\emph{a}ware \emph{l}ogical \emph{f}ramework for studying quantitative aspects of functional programs with effects.
  Like \calf{}, the language is based on an internal \emph{phase distinction} between the \emph{behavior} of a program and its \emph{cost} measured by an effect.
  \decalf{} extends \calf{} by accommodating other \emph{effects}, such as probabilistic choice, which requires a reformulation of \calf{}'s approach to cost analysis: rather than rely on a separable notion of cost, here \emph{a cost bound is simply another program}.
  Formally, every type is equipped with an intrinsic preorder, allowing effectful programs to be compared for cost inequality.
  This approach serves as a streamlined alternative to the standard method of isolating a cost recurrence and readily extends to higher-order, effectful programs.

  The development proceeds by first introducing the \decalf{} type system, which is based on an intrinsic cost ordering among terms that restricts in the behavioral phase to extensional equality.
  This formulation is then applied to illustrative examples, including pure and effectful sorting algorithms.
  Finally, \decalf{} is semantically justified via a model in the topos of augmented simplicial sets.
\end{abstract}

\maketitle

\section{Introduction}\label{sec:intro}

The \calf{} language~\citep{niu-sterling-grodin-harper:2022} is a dependent type theory that consolidates the specification and verification of the \emph{behavior} and \emph{cost} of programs.
For example, in \calf{} it is not only possible to prove that insertion sort and merge sort are behaviorally equal, but also they perform at most $n^2$ and at most $n \lg n$ comparisons, respectively, on an input list of length $n$.

It may seem, at first, that the stated properties are incompatible---if the sorting functions are \emph{equal} as functions on lists, then how can they have \emph{different} cost properties?
Moreover, what does it mean, type-theoretically, for the two algorithms to perform the stated number of comparisons?
The key to understanding how these questions are handled by \calf{} lies in the combination of two recent developments in type theory:
\begin{enumerate}
    \item The view of cost counting as a computational effect, implemented equationally in a dependent variant of Levy's \emph{call-by-push-value}~\citep{levy:2003:book,kavvos-morehouse-licata-danner:2019}.
    \item The reformulation of \emph{phase distinctions}~\citep{harper-mitchell-moggi:1990} in terms of modalities, emanating from Sterling's Synthetic Tait Computability~\citep{sterling:2021:thesis,sterling-harper:2021}.
\end{enumerate}

\subsection{Call-By-Push-Value}

In call-by-push-value, \emph{value} types $X,Y,Z$, which classify pure data, are distinguished from \emph{computation} types $A,B,C$, which classify effectful programs.%
\footnote{Although not all models of call-by-push-value are of this form, it is instructive to think of computation types as being the algebras for a strong monad on the category of value types and pure functions.}
Following the formulation of \calf{}, these two classes of types are defined in the following signature in an extensional, higher-order, dependently typed \emph{logical framework} (LF):%
\begin{center}
  \begin{minipage}{0.4\columnwidth}
    \iblock{
      \mrow{\tpv,\tpc : \jdg}
      \mrow{\mathsf{val} : \tpv \to \jdg}
      \mrow{\mathsf{comp} : \tpc\to\jdg}
      \mrow{\tmc{A} \eqdef \tmv{\U{A}}}
    }
  \end{minipage}%
  \begin{minipage}{0.6\columnwidth}
    \iblock{
      \mrow{\mathsf{U} : \tpc \to \tpv}
      \mrow{\mathsf{F} : \tpv \to \tpc}
      \mrow{\mathsf{ret} : \tmv{X} \to \tmc{\F{X}}}
      \mrow{\mathsf{bind} : \tmc{\F{X}} \to {(\tmv{X} \to \tmc{A})} \to \tmc{A}}
    }
  \end{minipage}%
\end{center}
Value types $\isof{X}{\tpv}$ and computation types $\isof{A}{\tpc}$ each have a corresponding collection of terms, $\tmv{X}$ and $\tmc{A}$.
Following \citet{niu-sterling-grodin-harper:2022}, we define computations as $\tmc{A} \coloneqq \tmv{\U{A}}$, leading to a less bureaucratic version of call-by-push-value in which the suspension and forcing operations are invisible.
The two classes of types are linked by adjoint type formers: $\F{X}$ describes effectful computations that may return a value of type $X$, and $\U{A}$ describes suspensions of effectful computations of type $A$.
The \textsf{ret} and \textsf{bind} constructs return values and sequence computations.

\subsection{Compositional Cost Analysis in Calf}\label{sec:intro:calf}
The purpose of separating values from computations is to give a compositional account of cost.
In particular, \calf{} instruments code with its cost by means of a write-only \emph{cost counting} effect $\mstept{A}{c}{e}$ that annotates a computation $e : A$ with an additional $c : \costty$ units of cost, so determining the figure-of-merit for cost analysis.
(In the case of sorting, the comparison operation is instrumented with this effect.)
The value type $\costty$ is equipped with the structure of a monoid that is respected by the cost effect. %
\begin{center}
  \begin{minipage}{0.4\columnwidth}
    \iblock{
      \mrow{\costty : \tpv}
      \mrow{0 : \tmv{\costty}}
      \mrow{+ : \tmv{\costty} \to \tmv{\costty} \to \tmv{\costty}}
      \mrow{\_ : \mathsf{isMonoid}(\costty, 0, +)}
    }
  \end{minipage}%
  \begin{minipage}{0.6\columnwidth}
    \iblock{
      \mrow{\mathsf{charge} : \tmv{\costty} \to {\tmc{A} \to \tmc{A}}}
      \mrow{\_ : \mstep{0}{a} = a}
      \mrow{\_ : \mstep{c_1}{\mstep{c_2}{a}} = \mstep{c_1 + c_2}{a}}
    }
  \end{minipage}%
\end{center}

When the insertion sort algorithm $\textit{isort} : \listty{\nat} \rightharpoonup \F{\listty{\nat}}$ is appropriately annotated with this effect, its cost may be shown in \calf{} to be upper bounded by quadratic comparisons by proving
\begin{align}\label{eq:isort-eq}
  \textit{isort} = \lambda l.\ \mstep{r(l)}{\ret{\underline{\textit{sort}}~l}},
\end{align}
for some (pure) function $r : \listty{\nat} \to \costty$ with $r(l) \le \len{l}^2$ for all lists $l$, where the unique (pure) sorting function $\underline{\textit{sort}} : \listty{\nat} \to \listty{\nat}$ specifies the behavior of $\textit{isort}$.%
\footnote{In many examples, we implicitly assume and apply a monoid homomorphism $(\nat,0,+) \to (\costty,0,+)$, treating costs as numbers that are combined additively.}
The reliance on extracting a function $r$ reflects a long-standing tendency in the literature to isolate a ``mathematical'' characterization of the cost of executing an algorithm by a function---typically recursively defined, then called a \emph{recurrence}---that defines the cost incurred as a function of its input.

In this way, the \calf{} type theory is capable of expressing and verifying cost bounds on algorithms.
But now, when the cost counting effect is included, insertion sort and merge sort are \emph{not equal}, exactly because they have different costs. %
How might we speak of their behavioral equivalence?

\subsection{The Cost--Behavior Phase Distinction}\label{sec:intro:phase-distinction}
Here the second key idea comes into play, the introduction of a \emph{phase distinction} between cost and behavior, which is a proposition $\ExtOpn$ representing a \emph{switch} that collapses all cost information when activated/assumed.%
\footnote{The logic of \calf{} is \emph{intuitionistic}, so an indeterminate proposition need not be either true or false in a given scope.}
\begin{center}
  \begin{minipage}{0.5\columnwidth}
    \iblock{
      \mrow{\ExtOpn : \jdg}
      \mrow{\_ : \impl{\isof{p,p'}{\ExtOpn}} p = p'}
    }
  \end{minipage}%
  \begin{minipage}{0.5\columnwidth}
    \iblock{
      \mrow{\ext{\mathcal{J}} \eqdef \ExtOpn \to \mathcal{J}}
      \mrow{\_ : \ext{\tmv{\costty} \cong \unit}}
    }
  \end{minipage}%
\end{center}
When $\ExtOpn$ is true/inhabited, \eg by assumption in the context, then \emph{the cost counting effect is erased}, and hence the two sorting algorithms are deemed \emph{behaviorally equal}.
This is achieved by requiring that the cost model $\costty$ is \emph{algorithmic}%
, meaning that it is isomorphic to the unit type in the behavioral phase.
This means that in the behavioral phase, every $c : \costty$ is equal to $0$, and thus $\mstep{c}{e} = \mstep{0}{e} = e$.
The \emph{behavioral modality} $\Op$ governs \emph{behavioral} specifications in the sense that any type of the form $\ext{A}$ is oblivious to cost.
In general, programs, equations, and inequalities in \decalf{} take cost structure into account.
To consider only cost-ignoring behavioral properties, we study programs in the fragment of \decalf{} under the behavioral phase; for example, insertion sort and merge sort are behaviorally equal, $\ext{\textit{isort} = \textit{msort}}$.

\subsection{Compositional Cost Analysis with Effects}\label{sec:intro:cost-analysis}
To motivate the contributions of this paper, we turn our attention to effects beyond cost.
As a representative example, consider a probabilistic choice algebraic effect, axiomatized in \cref{fig:effect-probability}.~\citep{plotkin-power:2002}
The computation $\mflip{p}{a_0}{a_1}$ computes $a_1$ with probability $p$ and computes $a_0$ otherwise.

\begin{figure}
  \iblock{
    \mrow{\mathsf{flip} : \QQ \to \tmc{A} \to \tmc{A} \to \tmc{A}}
    \mrow{\_ : \mflip{0}{a_0}{a_1} = a_0}
    \mrow{\_ : \mflip{p}{a_0}{a_1} = \mflip{1-p}{a_1}{a_0}}
    \mrow{\_ : \mflip{p}{a}{a} = a}
    \mhang{
      \_ :
      \prn{p = 1 - ((1 - pq)(1 - r))}
    }{
      \mrow{
        \to
        \mflip{pq}{\mflip{r}{a_0}{a_1}}{a_2} = \mflip{p}{a_0}{\mflip{q}{a_1}{a_2}}
      }
    }
    \mrow{\_ : \mstep{c}{\mflip{p}{a_0}{a_1}} = \mflip{p}{\mstep{c}{a_0}}{\mstep{c}{a_1}}}
  }
  \caption{
    Specification of the $\mathsf{flip}$ primitive for finitary probabilistic choice, which describes a convex space.
  }
  \label{fig:effect-probability}
\end{figure}

Using this effect, the randomized quicksort algorithm~\citep{hoare:1961,hoare:1962} may be defined as $\textit{qsort} : \listty{\nat} \rightharpoonup \F{\listty{\nat}}$.
Verifying that the behavior of quicksort is equivalent to that of insertion sort, $\ext{\textit{isort} = \textit{qsort}}$, is straightforward.
How, though, might the cost of quicksort be upper bounded by, say, its worst-case quadratic number of comparisons?
While the method of recurrences is desirable when applicable, it cannot be expected in general that the cost can be so characterized, such as in the presence of effects.
There is no (pure) recurrence $r : \listty{\nat} \to \costty$ such that
\[ \textit{qsort} = \lambda l.\ \mstep{r(l)}{\ret{\underline{\textit{sort}}~l}}, \]
since the cost of running $\textit{qsort}$ on a list $l$ is, by design, random!

To address this issue, we propose to extend \calf{} with a notion of \emph{program-level cost inequality}, such that the proposition ``the worst-case cost of $\textit{qsort}$ is quadratic'' can be formally rendered as
\begin{align}\label{eq:qsort-leq}
  \textit{qsort} \mathrel{\Alert{\le}} \lambda l.\ \mstep{\len{l}^2}{\ret{\underline{\textit{sort}}~l}}
\end{align}
in dependent type theory, allowing effectful program $\textit{qsort}$ to be compared directly to an effectful program, treated as a cost bound.
The central principle that drives the present work is this:
\begin{center}
    \textit{What better way to define the cost of an effectful program \\ than by another effectful program?}
\end{center}
Using this principle, we provide a synthetic account of cost analysis for effectful programs that extends the established \calf{} methodology beyond the special case of pure functional programs.

\subsection{\decalf{}: a Directed, Effectful Cost-Aware Logical Framework}\label{sec:intro:decalf}

The key to achieving these goals is to reformulate the \calf{} type theory to consider \emph{inequalities}, as well as equalities, on programs.
Every type will come equipped with a preorder saying, intuitively, that $e \le e'$ when the cost of $e$ is upper bounded by the cost of $e'$, and the behavior of $e$ and $e'$ agree.
\begin{enumerate}
  \item Reflexivity means that it is always valid to say, in effect, that a piece of code ``costs what it costs.''  However, it of course is desirable to characterize the cost of a program more succinctly and informatively using, say, a closed form.
  \item Transitivity means that bounds on bounds may be consolidated, facilitating modular reasoning.
\end{enumerate}
Thus, \decalf{} is a \emph{directed} extension of \calf{}, where the ``directed'' aspect is manifested by the inequational judgments.
Inequalities, like equalities, account for both cost and behavior.
\begin{enumerate}
  \item Algorithmically (\ie, in general), inequality relaxes costs: if $\leVal{c}{c'}{\costty}$, then $\leCmp[A]{\mstep{c}{e}}{\mstep{c'}{e}}$.
  \item Behaviorally (\ie, assuming $\ExtOpn$), inequality is just equality: if $\leCmp[A]{e}{e'}$, then $\Op\prn{\eqCmp[A]{e}{e'}}$.
\end{enumerate}
Program inequality compares both the cost (for inequality) and the behavior (for equality) of a program simultaneously.
In general, this is essential: the \emph{cost} of later computations may depend on the \emph{behavior} of earlier ones.
Thus, we see that the cost--behavior phase distinction counterposes the \emph{interference} of behavior with cost to the \emph{noninterference} of cost with behavior \citep[\S 2.7]{niu-sterling-grodin-harper:2022}.

In the pure case, with cost as the only effect, the data of the original \calf{} methodology suffices to prove a program-level inequality, using the former condition to relax the true recurrence to the desired upper bound.
For example, to prove that
\[ \textit{isort} \le \lambda l.\ \mstep{\len{l}^2}{\ret{\underline{\textit{sort}}~l}}, \]
it suffices to combine \cref{eq:isort-eq} with reflexivity of the preorder and the proof that $r(l) \le \len{l}^2$.
In the presence of other effects, this inequality-based method of cost upper bounds generalizes that of \calf{}, such as in the case of $\textit{qsort}$: since all possible executions use at most $\len{l}^2$ cost, we will upper bound $\textit{qsort}$ as in \cref{eq:qsort-leq}.

\paragraph{Synopsis}
In \cref{sec:decalf-type-theory} we define the \decalf{} type theory;
in \cref{sec:examples} we formulate algorithms (some with effects) and derive their cost bounds;
and in \cref{sec:semantics} we justify \decalf{} topos-theoretically.

\subsection{Related Work}\label{sec:related-work}
As regards related work on cost analysis, the principal reference is \calf{}~\citep{niu-sterling-grodin-harper:2022}, on which the present work is based; \calf{} was itself inspired in part by prior works on denotational cost analysis via monads~\citep{danner-licata-ramyaa:2015} and call-by-push-value~\citep{kavvos-morehouse-licata-danner:2019}.
The \decalf{} type theory with its built-in inequality relation is closely related to the idea of \emph{directed type theory}~\citep{licata-harper:2011,riehl-shulman:2017}, which generalizes Martin-L\"of's identity types to account for directed identifications. Another important input to the design of \decalf{} is synthetic domain theory~\citep{hyland:1991,phoa:1991}, in which types are also equipped with an intrinsic preorder. Both of these inputs can be seen in our presheaf model of \decalf{}~(\cref{sec:semantics}), which resembles \emph{both} traditional models of higher category theory and directed type theory \emph{and} (pre)sheaf models of synthetic domain theory~\citep{fiore-rosolini:2001,fiore-rosolini:1997:cpos}. Our method to isolate presheaves that behave like preorders via orthogonality comes from \citet{fiore:1997}, whose ideas we have combined with the modern accounts of orthogonal reflection of \citet{rijke-shulman-spitters:2020} and \citet{christensen-opie-rijke-scoccola:2020}.

\section{The \decalf{} Type Theory}\label{sec:decalf-type-theory}

The \decalf{} type theory is a dependent extension of Levy's call-by-push-value framework~\citep{levy:2003:book} in which types are classified into \emph{value types} $X,Y,Z$ (such as suspensions $\U{A}$, $\nat$, and $\listty{X}$) and \emph{computation types} $A,B,C$ (such as $\F{X}$ and $X \rightharpoonup A$).
As with \calf{}, the \decalf{} type theory includes a \emph{phase distinction} between \emph{behavioral} and \emph{algorithmic} aspects of programs.
The hallmark of \decalf{} is inequality on programs, defined at all value types $X$:
\par\iblock{
  \mrow{{\leq} : \tmv{X} \to \tmv{X} \to \jdg}
  \mrow{\_ : \mathsf{isPreorder}(X, \leq)}
}\noindent
By the definition of $\tmc{A}$, computations of type $A$ are compared at value type $\U{A}$.
The approximation preorder on \decalf{} values and computations is induced by these principles:
\begin{enumerate}
  \item All functions are (automatically) monotone.
  \item Functions are compared pointwise.
  \item Under the behavioral phase, inequality implies equality.
\end{enumerate}
The behavioral requirement expresses that the approximation ordering is solely to do with cost: when cost effects are suppressed, the preorder is just equality, and thus has no effect on the behavior of the program.
We render these conditions formally in the LF:
\par\iblock{
  \mrow{\_ : (f:\tmv{X}\to\tmv{Y}) \to\leVal{x}{x'}{X} \to \leVal{f(x)}{f(x')}{Y}}
  \mrow{\_ : ((x:\tmv{X})\to \leCmp[Y(x)]{f(x)}{f'(x)}) \to \leVal{\textsf{lam}(f)}{\textsf{lam}(f')}{\Pi^+(X;Y)}}
  \mrow{\_ : \ext{\leVal{x}{x'}{X} \to \eqVal{x}{x'}{X}}}
}\noindent
We may also internalize this judgmental structure as a value type:
\par\iblock{
  \mrow{\mathsf{leq} : \tmv{X} \to \tmv{X} \to \tpv}
  \mrow{\_ : \tmv{\mathsf{leq}_X(x,x')} \cong \prn{\leVal{x}{x'}{X}}}
}\noindent
For readability, we will continue to discuss judgmental inequalities, although all stated cost bounds are provable internally.

\paragraph{Cost Ordering}
Since every type is intrinsically equipped with a preorder and all functions are guaranteed to be monotone, this automatically makes $(\costty, \le_{\costty}, +, 0)$ a preordered monoid.
Therefore, as a consequence of monotonicity, $\leCmp[A]{\mstept{A}{c}{e}}{\mstept{A}{c'}{e}}$ follows from $\leVal{c}{c'}{\costty}$ for any computation $\isof{e}{A}$.
This is the means by which cost bounds are relaxed in an analysis.
\section{Verification Examples}\label{sec:examples}

Equipped with equality and inequality of programs, we now provide examples of how a cost analysis may be performed.
In place of a cost bound, we simply use another program.
The purpose of cost analysis, then, will be to condense the details of a complex program.

In the forthcoming examples, we will instantiate the cost model to $(\omega,
\le_\omega, +, 0)$, a synthetic version of the natural numbers equipped with their standard inequality relation and additive monoid structure. Note that $\omega$ is not the inductive type of natural
numbers, which is discrete as a preorder. We note that $\omega$ is an algorithmic
type and implement it as a \emph{quotient
inductive type} \citep{kaposi-kovacs-altenkirch:2019} in
\cref{sec:cost-structure-semantics}.

\subsection{Pure Algorithms}

First, we discuss pure algorithms in which cost is the only available effect, providing an inequality-based presentation of the techniques of \calf{}~\citep{niu-sterling-grodin-harper:2022}.

\begin{figure}
  \iblock{
    \mrow{\textit{insert} : \nat \rightharpoonup \listty{\nat} \rightharpoonup \F{\listty{\nat}}}
    \mrow{\textit{insert}~x~\nilex = \ret{\consex{x}{\nilex}}}
    \mhang{\textit{insert}~x~\prn{\consex{y}{ys}} =}{
      \mrow{\bindex{\mstep{1}{x \le^? y}}{b}}
      \mhang{\boolif{b}{\ret{\consex{x}{\consex{y}{ys}}}}{}}{
        \mrow{\bindex{\textit{insert}~x~ys}{ys'}\ret{\consex{y}{ys'}}}
      }
    }
  }
  \smallskip
  \iblock{
    \mrow{\textit{isort} : \listty{\nat} \rightharpoonup \F{\listty{\nat}}}
    \mrow{\textit{isort}~\nilex = \ret{\nilex}}
    \mrow{\textit{isort}~\prn{\consex{x}{xs}} = \bindex{\textit{isort}~xs}{xs'}\textit{insert}~x~xs'}
  }
  \caption{%
    Insertion sort, with one cost per comparison. %
  }
  \label{code:isort}
\end{figure}

\begin{example}[Insertion Sort]
  Consider the implementation of $\textit{isort}$ in \cref{code:isort}.
  Here, we count the number of comparison operations performed (all in the subroutine $\textit{insert}$) by instrumenting them with the cost effect.
  The cost incurred by the computation $\textit{isort}~l$ depends on the particular elements of the list $l$.
  Thus, for the most precise upper bound of $\textit{isort}$, one could of course takes $\textit{isort}$ as its own upper bound, by reflexivity: $\textit{isort} \le \textit{isort}$.
  However, this bound provides more detail than a client may wish for.
  Rather than characterize this cost precisely, then, it is common to give only an upper bound.
  In the worst case, $\textit{isort}~l$ incurs no more than $\len{l}^2$ cost, $l$ is in reverse-sorted order.
  Thus, we may define
  \[ \bound{\textit{isort}} \coloneq \lambda l.\, \mstep{\len{l}^2}{\ret{\underline{\textit{sort}}~l}} \]
  as an upper bound for $\textit{isort}$, as described in \cref{sec:intro:decalf}.
  To prove that $\textit{isort} \le \bound{\textit{isort}}$, one can either
  \begin{enumerate}
    \item adapt the recurrence approach described in \cref{sec:intro:calf}; or
    \item prove this fact directly, using the fact that inequality of functions is pointwise, by induction on an input list $l$.
  \end{enumerate}
  In terms of cost, this verification shows that $\textit{isort}~l$ makes at most $\len{l}^2$ comparisons.
  Moreover, in terms of behavior, this inequality states that $\textit{isort}$ is a correct sorting algorithm: the equality
  \[ \ext{\textit{isort} = \lambda l.\, \ret{\underline{\textit{sort}}~l}} \]
  follows immediately from the cost-and-behavior inequality, using the fact that inequality is behaviorally just equality and cost is behaviorally erased.
  Thus, the proof that $\textit{isort} \le \bound{\textit{isort}}$ constitutes a proof of both the cost and correctness of $\textit{isort}$.

  Although it is less common than proving an upper bound, it is also possible to prove a lower bound on the cost of a computation.
  Since $\textit{insert}~x~l$ costs at least $1$ on a non-empty list $l$ and $\textit{isort}$ is length-preserving, $\textit{isort}~l$ incurs at least $\len{l} - 1$ cost:
  \[ \leCmp{\lambda l.\, \mstep{\len{l} - 1}{\ret{\underline{\textit{sort}}~l}}}{\textit{isort}}. \qedhere \]
\end{example}

Adapting the work of \citet{niu-sterling-grodin-harper:2022}, we may similarly define the merge sort algorithm, \textit{msort}, and prove that it is bounded by cost $\len{l}\lg\len{l}$ and behaviorally coheres with $\underline{\textit{sort}}$:
\[ \leCmp{\textit{msort}}{\lambda l.\, \mstep{\len{l}\lg\len{l}}{\ret{\underline{\textit{sort}}~l}}} \eqcolon \bound{\textit{msort}}. \]
Using the fact that program inequality is behaviorally equality, we may recover the proof that these two sorting algorithms are behaviorally equal from their respective cost-and-behavior inequalities.

\begin{theorem}
  In the behavioral phase, $\eqCmp{\textit{isort}}{\textit{msort}}$.
\end{theorem}
\begin{proof}
  Behaviorally, the inequalities in the cost bounds are equalities and the cost operation is trivialized: \[ \textit{isort}
  = \bound{\textit{isort}}
  = \lambda l.\, \ret{\underline{\textit{sort}}~l}
  = \bound{\textit{msort}}
  = \textit{msort}. \qedhere \]
\end{proof}

\subsection{Effectful Algorithms}\label{sec:effectful-algorithms}

Writing cost bounds using program inequalities, we can extend \decalf{} with various computational effects, such as probabilistic choice~(\cref{fig:effect-probability}), and prove bounds on effectful programs.

\begin{figure}
  \iblock{
    \mrow{\textit{choose} : \prodty{\nat}{\listty{\nat}} \rightharpoonup \F{\prodty{\nat}{\listty{\nat}}}}
    \mrow{\textit{choose}~\prn{x,\nilex} = \ret{x, \nilex}}
    \mhang{\textit{choose}~\prn{x,\consex{y}{ys}} =}{
      \mhang{\mflipp{1 / (2 + \len{ys})}}{
        \mrow{({\bindex{\textit{choose}~\prn{y,ys}}{\prn{\pair{pivot}{l}}}\ret{\pair{pivot}{\consex{x}{l}}}},}
        \mrow{\phantom{(}{\ret{\pair{x}{\consex{y}{ys}}}})}
      }
    }
  }
  \smallskip
  \iblock{
    \mrow{\textit{partition} : \nat \rightharpoonup \listty{\nat} \rightharpoonup \F{\prodty{\listty{\nat}}{\listty{\nat}}}}
    \mrow{\textit{partition}~pivot~\nilex = \ret{\pair{\nilex}{\nilex}}}
    \mrow{\textit{partition}~pivot~\prn{\consex{x}{xs}} =}
    \mrow{\quad\bindex{\textit{partition}~pivot~xs}{\prn{\pair{xs_1}{xs_2}}}}
    \mrow{\quad\bindex{\mstep{1}{x \le^? pivot}}{b}}
    \mrow{\quad\boolif{b}{\ret{\pair{\consex{x}{xs_1}}{xs_2}}}{\ret{\pair{xs_1}{\consex{x}{xs_2}}}}}
  }
  \smallskip
  \iblock{
    \mrow{\textit{qsort} : \listty{\nat} \rightharpoonup \F{\listty{\nat}}}
    \mrow{\textit{qsort}~\nilex = \ret{\nilex}}
    \mrow{\textit{qsort}~\prn{\consex{x}{xs}} =}
    \mrow{\quad\bindex{\textit{choose}~\prn{x, xs}}{\prn{\pair{pivot}{l}}}}
    \mrow{\quad\bindex{\textit{partition}~pivot~l}{\prn{\pair{l_1}{l_2}}}}
    \mrow{\quad\bindex{\textit{qsort}~l_1}{l_1'}}
    \mrow{\quad\bindex{\textit{qsort}~l_2}{l_2'}}
    \mrow{\quad\ret{l_1' \mdoubleplus \consex{pivot}{l_2'}}}
  }
  \Description{Quicksort algorithm}
  \caption{
    Quicksort algorithm, where the \textit{choose} auxiliary function chooses a pivot uniformly randomly.
    As in \cref{code:isort}, the cost instrumentation tracks one unit of cost per comparison.
  }
  \label{code:qsort}
\end{figure}

\begin{example}[Quicksort]
  In \cref{code:qsort}, we define a variant of the quicksort algorithm \citep{hoare:1961,hoare:1962} in which the pivot is chosen uniformly randomly.
  The number of comparisons performed by $\textit{qsort}~l$ depends on which element is chosen as a pivot; in the worst case, it can perform quadratically many comparisons.
  This cost bound, as formulated in \cref{sec:intro:cost-analysis}, may be proved by induction:
  \[ \leCmp{\textit{qsort}}{\lambda l.\, \mstep{\len{l}^2}{\ret{\underline{\textit{sort}}~l}}}. \]
  The fact that $\textit{qsort}~l$ is randomized is not reflected in this bound: regardless of the chosen pivot, it always incurs at most $\len{l}^2$ cost and returns $\underline{\textit{sort}}~l$.
  In other words, the use of randomization was \emph{benign}: behaviorally, it is invisible, rendering the program effect-free.
  In particular, it is the case that $\ext{\textit{isort} = \textit{msort} = \textit{qsort}}$.
\end{example}

By including non-cost effects in the upper bounding program, it is also possible to analyze the effects used in a program (\eg, the specific distributional costs), which is of particular importance when analyzing programs that use effects in a non-benign manner.

\subsection{Higher-Order Functions}\label{sec:hof}

We now turn our attention to \emph{higher-order} functions that take suspended computations as input.

\begin{figure}
  \iblock{
    \mrow{\textit{map} : \U{X \rightharpoonup \F{Y}} \rightharpoonup \listty{X} \rightharpoonup \F{\listty{Y}}}
    \mrow{\textit{map}~f~\nilex = \ret{\nilex}}
    \mrow{\textit{map}~f~\prn{\consex{x}{xs}} =}
    \mrow{\quad\bindex{\textit{map}~f~xs}{ys}}
    \mrow{\quad\bindex{f~x}{y}}
    \mrow{\quad\ret{\consex{y}{ys}}}
  }
  \caption{
    Implementation of the \textit{map} function on lists, which applies a suspended function elementwise to an input list.
  }
  \Description{
    Implementation of the \textit{map} function on lists, which applies a suspended function elementwise to an input list.
    No cost is instrumented explicitly, but the applications of $f$ may incur cost (and/or other effects).
  }
  \label{code:map}
\end{figure}

\begin{example}[List map]\label{ex:list-map}
  Consider the list map function in \cref{code:map}.
  If nothing is known about the input $f$, then \textit{map} is the only reasonable bound for itself---to be sure, if $f$ made use of order-dependent effects, such as printing, the exact implementation of $\textit{map}$ cannot be simplified.
  However, if some properties about $f$ are known, it is possible to prove a more concise bound.
  \begin{enumerate}
    \item
      Suppose $f$ is known to be constant-cost, meaning that there exists a pure $\underline{f} : X \to Y$ and $c : \costty$ such that for all $x : X$,
      \[ f~x \le \mstep{c}{\ret{\underline{f}~x}}. \]
      Then it is the case that $\textit{map}~f$ is linear-cost,
      \[ \textit{map}~f \le \lambda l.\ \mstep{c\len{l}}{\ret{\underline{\textit{map}}~\underline{f}~l}}, \]
      where $\underline{map}$ is the pure list map function.

    \item
      Alternatively, suppose $f$ is known to have cost bounded by a Bernoulli distribution,
      \[ f~x \le \mflip{p}{\ret{\underline{f}~x}}{\mstep{c}{\ret{\underline{f}~x}}}. \]
      Then $\textit{map}~f$ has cost bounded by a $\len{l}$-trial binomial distribution, using $\textit{binomial} : \nat \rightharpoonup \F{\unit}$ (defined by induction):
      \[ \textit{map}~f \le \prn{\lambda l.\ \bindex{\textit{binomial}~\len{l}}{\_}{\ret{\underline{\textit{map}}~\underline{f}~l}}}. \qedhere \]
  \end{enumerate}
\end{example}

This style of reasoning aligns well with existing informal techniques.
If details about an input computation are known, then a concise and insightful bound can be derived.
Otherwise, one must examine the program in its entirety to understand the behavior.
\DeclareFontFamily{U}{min}{}
\DeclareFontShape{U}{min}{m}{n}{<-> udmj30}{}
\newcommand\yo{\!\text{\usefont{U}{min}{m}{n}\symbol{'210}}\!}

\NewDocumentCommand\SPX{}{\mathbf{\Delta}}
\NewDocumentCommand\AugSPX{}{\mathbf{\Delta}_\bot}
\NewDocumentCommand\OpCat{m}{#1\Sup{\mathsf{op}}}
\NewDocumentCommand\Yo{o}{\yo\IfValueT{#1}{\Sub{#1}}}
\NewDocumentCommand\CAT{}{\mathbf{Cat}}
\NewDocumentCommand\PREORD{}{\mathbf{Preord}}
\NewDocumentCommand\II{}{\mathbb{I}}
\NewDocumentCommand\RxP{}{\mathsf{P}}
\NewDocumentCommand\pathrel{m}{\sqsubseteq\Sub{#1}}
\NewDocumentCommand\specrel{m}{\preceq\Sub{#1}}
\NewDocumentCommand\irel{}{\le\Sub{\II}}
\NewDocumentCommand\Hom{mmm}{\hom\Sub{#1}\prn{#2,#3}}

\NewDocumentCommand\UU{}{\mathcal{U}}
\NewDocumentCommand\VV{}{\UU^{+}}

\section{A Model of \decalf{}}\label{sec:semantics}

Our goal is to construct a model of type theory that contains a nontrivial interpretation of the constructs of \decalf{}: this must necessarily contain a universe of types equipped with a built-in preorder structure, as well as a phase distinction proposition --- such that in the behavioral phase, the inequality relations collapse to equalities.
This will be achieved by first compiling \decalf{} to extensional dependent type theory with some additional axioms and then providing a model that validates these axioms.

\begin{definition}[Directed interval]
  A \emph{directed interval} is a type $\II$ equipped with two constants $0,1 : \II$ and a preorder $\irel$ such that for all $i : \II$, $0 \irel i$ and $i \irel 1$.%
  \footnote{We do \emph{not} assume here that ${0\not=1}$, as is done by \citet{riehl-shulman:2017}.}
\end{definition}

\begin{mainidea}[Paths for automatic monotonicity]\label{idea:interval}
  The first problem to solve when building a model of \decalf{} is to devise a binary relation $\prn{\pathrel{X}}\subseteq X\times X$ on every type $X$ such that any function $X \to Y$ is \emph{automatically} monotone. The solution to this problem, first discovered in the world of synthetic domain theory, is to define $\prn{\pathrel{X}}$ uniformly in $X$ by considering functions into $X$ from a directed interval.
  Such an interval always induces a \emph{path relation}
  \[ x \pathrel{X} x' \eqdef{} \exists \alpha \colon \II\to X.\ (\alpha0 = x) \land (\alpha1 = x'), \]
  that is automatically preserved by any function $f\colon X \to Y$.
  To prove this, suppose $x\pathrel{X} x'$, and we will show that $fx\pathrel{Y} fx'$.
  By definition, we may assume some path $\alpha\colon \II\to X$ such that $\alpha0 = x$ and $\alpha1=x'$; then the composite $\beta \eqdef f\circ \alpha \colon\II\to Y$ satisfies $\beta0 = fx$ and $\beta1 = fx'$, and so $fx\pathrel{Y} fx'$.
\end{mainidea}

Although the idea of an interval lets us define a reflexive binary relation on every type $X$, this relation does not enjoy almost any of the properties that we need in order to model \decalf{}:

\begin{enumerate}
  \item \emph{Behavioral discreteness:} Assuming $\ExtOpn$, it need not be the case that $x \pathrel{X} x'$ implies $x = x'$.
  \item \emph{Path transitivity:} It need not be the case that $\prn{\pathrel{X}}$ is transitive, \ie exhibit $X$ as a preorder.
  \item \emph{Pointwise order:} It need not be the case that functions are ordered pointwise, \ie we do not necessarily have ${f \pathrel{X\to Y} f'}$ if and only if $\forall x : X.\ fx\pathrel{Y} f' x$.
\end{enumerate}

We can solve the first problem by defining the behavioral phase proposition $\ExtOpn$ in terms of the interval type.

\begin{mainidea}[Behavioral phase by equality of endpoints]\label{idea:behavioral-discreteness}
  We can force $x\pathrel{X} x'$ to imply $x = x'$ in the behavioral phase by defining $\ExtOpn$ to be the proposition $0 =_\II 1$.
  To see that this condition suffices, assume the behavioral phase holds, \ie $0 =_\II 1$.
  Because a proof $x \pathrel{X} x'$ consists of a map $\alpha : \II \to X$ with $\alpha0 = x$ and $\alpha1 = x'$, we have $x = \alpha0 = \alpha1 = x'$ as desired.
\end{mainidea}

In order to achieve transitivity and pointwise comparison of functions, we restrict our attention to a class of types, isolated in a \emph{reflective subuniverse}~\citep{rijke-shulman-spitters:2020}, that have these properties.

\begin{mainidea}[Transitivity and pointwise functions by orthogonality]\label{idea:transitivity-by-orthogonality}
  We can isolate the subuniverse of types $\UU_\RxP \subseteq \UU$ for which the path relation is transitive and pointwise on functions, where the corresponding reflection operator $\RxP : \UU \to \UU_\RxP$ is constructed as a quotient inductive type. The reflection operator is left adjoint to the inclusion $\UU_\RxP\hookrightarrow\UU$  and therefore restricts to an idempotent monad on $\UU$.
\end{mainidea}

Taking stock, what exactly do we need to do in order to construct a model of \decalf{} along the lines of \cref{idea:interval,idea:behavioral-discreteness,idea:transitivity-by-orthogonality}?
In extensional dependent type theory with quotient inductive types, it will be sufficient to assume a directed interval $(\II, \irel, 0, 1)$.
From this data, the rest follows by the construction of a reflective subuniverse $\UU_\RxP$ from general considerations about the interval (\cref{sec:semantics:synpreord,sec:algebra-models}).
We will now develop these ideas further, culminating in the construction of a topos-theoretic model of \decalf{} supporting a sound and complete embedding from the category of preorders (\cref{sec:semantics:sound-complete,sec:model-aug}).

\subsection{Synthetic Preorders from an Interval}\label{sec:semantics:synpreord}

As in \cref{idea:transitivity-by-orthogonality}, we will construct a reflective subuniverse $\UU_\RxP$. This will be achieved using the notion of \emph{orthogonality}.

\begin{definition}[Orthogonality and suborthogonality]\label{idea:orthogonality}
  Let $X$ be a type, and let $f \colon U \to V$ be a function; the concept of orthogonality (resp. suborthogonality) is one way to make precise the idea that $X$ behaves ``as if'' the map $f$ were an isomorphism (resp. an epimorphism).
  We say that $X$ is \emph{orthogonal} (resp. \emph{suborthogonal}) to $f$ when the precomposition $\prn{-\circ f} \colon \prn{V\to X}\to \prn{U\to X}$ is required to be an isomorphism (resp. a monomorphism).
\end{definition}

It happens that the class of types orthogonal to a finite set of maps can be assembled into a reflective subuniverse $\UU_\RxP$ with a modal operator $\RxP$, assuming sufficiently powerful quotient and inductive types~\citep[\S 2]{rijke-shulman-spitters:2020}.
For example, a type is \emph{proposition} (subsingleton) iff it is orthogonal to the unique map ${{*} \colon \mathbf{2} \to \unit}$; the induced reflective subuniverse is the universe of mere propositions, and the induced modal operator is propositional truncation.
Therefore, in order to obtain the reflective subuniverse $\UU_\RxP$, it suffices to find orthogonality conditions that imply path transitivity (\cref{sec:semantics:spi:transitivity}) and pointwise ordering of functions (\cref{sec:semantics:spi:boundary-separation}).

\subsubsection{Path-Transitivity}\label{sec:semantics:spi:transitivity}
First, we describe a sufficient condition on a type $X$ that causes the path relation $\pathrel{X}$ to be transitive, adapting the following constructions from \citet{fiore-rosolini:1997:cpos}.

\NewDocumentCommand\In{o}{\mathsf{inj}\IfValueT{#1}{\Sub{#1}}}
\NewDocumentCommand\Prj{o}{\mathsf{proj}\IfValueT{#1}{\Sub{#1}}}

\begin{definition}\label{def:path-transitivity}
  A type $X$ is \emph{path-transitive} when it is orthogonal to the inclusion $\tau\colon \II\lor\II\hookrightarrow\II_2$, where $\II_2 = \Compr{(i,j)}{i\irel j}$:
  \[
    \begin{tikzpicture}[diagram]
      \SpliceDiagramSquare<sq/>{
        nw = \unit,
        ne = \II,
        sw = \II,
        se = \II\lor\II,
        north = 1,
        west = 0,
        se/style = pushout,
        south/node/style = upright desc,
        east/node/style = upright desc,
        south = \In[1],
        east = \In[0],
        north/style = embedding,
        west/style = embedding,
        east/style = embedding,
        south/style = embedding,
        height = 1.5cm,
      }
      \node[right = 2.5cm of sq/se] (I2) {$\II_2$};
      \draw[->,bend right=30] (sq/sw) to node[sloped,below] {$\lambda i. \prn{i,1}$} (I2);
      \draw[->,bend left=30] (sq/ne) to node[sloped,above] {$\lambda i. \prn{0,i}$} (I2);
      \draw[->,exists] (sq/se) to node[desc] {$\tau$} (I2);
    \end{tikzpicture}
  \]
\end{definition}

\begin{restatable}{lemma}{LemPathTransitivity}\label{lem:path-transitivity-gives-preorder}
  If a type $X$ is path-transitive, then $\pathrel{X}$ is transitive.
\end{restatable}
\begin{proof}
  Suppose that $x \pathrel{X} x'$ and $x' \pathrel{X} x''$ are witnessed by paths $\alpha \colon \II\to X$ and $\beta \colon \II\to X$, respectively, noting that ${\alpha0 = x}$, ${\alpha1 = \beta0 = x'}$, and ${\beta1 = x''}$.
  We wish to construct a proof ${x \pathrel{X} x''}$, \ie a path $\phi : \II \to X$ such that $\phi 0 = x$ and $\phi 1 = x''$. First, consider the map $\gamma \eqdef \brk{\alpha \mid \beta} \colon \II\lor\II\to X$, using the universal property of the pushout.
  Because $X$ is orthogonal to $\tau\colon \II\lor\II\hookrightarrow \II_2$,
  we have a unique lift $\hat{\gamma} \colon \II_2 \to X$.
  We define $\phi i \eqdef \hat{\gamma}\prn{i,i}$
  and compute:
  \begin{align*}
    \phi 0 & = \hat{\gamma}\prn{0,0} = \gamma(\In[0]~0) = \alpha 0 = x         \\
    \phi 1 & = \hat{\gamma}\prn{1,1} = \gamma(\In[1]~1) = \beta 1  = x'' \qedhere
  \end{align*}
\end{proof}

\subsubsection{Boundary Separation}\label{sec:semantics:spi:boundary-separation}

As it stands, there could be two distinct paths $\II\to X$ that have the same endpoints. We wish to isolate the types $X$ for which paths are uniquely determined by their endpoints; this property was dubbed \emph{boundary separation} by \citet{sterling-angiuli-gratzer:2019,sterling-angiuli-gratzer:2022} and $\Sigma$-separation by \citet{fiore-rosolini:2001}.

\begin{definition}
  A type $X$ is \emph{boundary separated} when it is \emph{suborthogonal} to the boundary inclusion $\brk{0 \mid 1} \colon \mathbf{2}\to \II$, meaning that any path $\alpha \colon \II \to X$ is determined by its endpoints $\alpha(0)$ and $\alpha(1)$.
\end{definition}

When a type $X$ is boundary separated, its path relation $x \pathrel{X} x'$ (defined in terms of mere existence) is equivalent to a dependent sum over a path, which provides the evidence necessary to prove that functions are compared pointwise.

\begin{lemma}\label{lem:bsep-sigma}
  If a type $X$ is boundary separated, then for all $x,x' \colon X$, \[ (x \pathrel{X} x') \cong \sigJ{\II \to X}[\alpha]{(\alpha 0 = x) \land (\alpha 1 = x')}. \]
\end{lemma}

\begin{restatable}{lemma}{LemBSepImpliesPointwise}\label{lem:bsep-implies-pointwise}
  Let $x:X\vdash Y x$ be a family of boundary separated types, and let $f,f' : \prn{x:X} \to Y x$ be a pair of dependent functions. Then $f \pathrel{\prn{x:X}\to Y x} f'$ iff for all $x:X$ we have ${fx \pathrel{Y x} f'x}$.
\end{restatable}
\begin{proof}
  The forward direction is trivial, using automatic monotonicity at application of $x$.
  In the backwards direction, suppose that we have $\piJ{X}[x]{fx\pathrel{Yx}f'x}$, witnessed via \cref{lem:bsep-sigma} by a family of paths $\alpha \colon \piJ{X}[x]{\II \to Yx}$.
  To obtain a path from $f$ to $f'$, it suffices to construct a path $\beta \colon \II \to \piJ{X}[x]{Yx}$ with boundaries $f$ and $f'$, which may be defined as $\beta~i~x \eqdef \alpha~x~i$.
\end{proof}

Although the presentation of boundary separation in terms of suborthogonality is simple and elegant, it will be advantageous to observe that boundary separation can also be seen as an orthogonality property so as to incorporate it into a reflective subuniverse. To that end, we introduce \emph{path suspensions} below in order to state \cref{lem:bsep-as-orthogonality}, which characterizes boundary separation in terms of orthogonality.

\begin{definition}
  We define the \emph{path suspension} of a type $X$ to be following pushout:
  \[
    \DiagramSquare{
      nw = X\times \mathbf{2},
      ne = X\times \II,
      sw = \mathbf{2},
      se = \mathbb{S}X,
      se/style = pushout,
      east/style = {->,exists},
      south/style=  {->,exists},
      west = \Prj[2],
      north = X\times\brk{0\mid 1},
      width = 2.85cm,
      height = 1.25cm,
    }
  \]
\end{definition}

\begin{restatable}{lemma}{LemBSepAsOrthogonality}\label{lem:bsep-as-orthogonality}
  A type $X$ is boundary separated iff it is orthogonal to the path suspension $\mathbb{S}{*}\colon \mathbb{S} \mathbf{2} \to \mathbb{S}\unit$ of the unique map ${*} \colon \mathbf{2}\to\unit$.
\end{restatable}
\begin{proof}
  A map $\mathbb{S}\mathbf{2}\to X$ is a pair of paths between two fixed elements of $X$; a map $\mathbb{S}\unit \cong \II \to X$ is a single path between two fixed elements. Thus, orthogonality to the suspension $\mathbb{S}\mathbf{2}\to\mathbb{S}\unit$ means precisely that any two paths between the same elements are equal, which is precisely suborthogonality to $\brk{0 \mid 1} \colon \mathbf{2}\to \II$.
\end{proof}

\subsubsection{Synthetic Preorders}

We now come to a suitable definition of ``synthetic preorder'' using these two notions.

\begin{definition}
  A type $X$ is called a \emph{synthetic preorder} when it is both path-transitive and boundary separated, \ie orthogonal to both $\tau \colon \II\lor\II\hookrightarrow\II_2$ and $\mathbb{S}{*}\colon\mathbb{S}\mathbf{2}\to\mathbb{S}\unit$.
\end{definition}

The benefit of defining synthetic preorders in terms of orthogonality is that they automatically form a \emph{reflective subuniverse} of types, $\UU_\RxP$, which is equipped with a \emph{localization} modal operator $\RxP$~\citep{rijke-shulman-spitters:2020} that computes a ``best approximation'' (\ie, a \emph{reflection}) of a type $X$ by a synthetic preorder $\RxP{X}$.

\begin{definition}[Reflective subuniverse of synthetic preorders] %
  Define $\UU_\RxP$ to be the universe of synthetic preorders.
  The associated modal operator ${\RxP : \UU \to \UU_\RxP}$ can be constructed by means of a quotient inductive type that localizes at $\tau \colon \II\lor\II\hookrightarrow\II_2$ and $\mathbb{S}{*}\colon\mathbb{S}\mathbf{2}\to\mathbb{S}\unit$.
\end{definition}

This construction further implies that the subuniverse $\UU_\RxP$ is cartesian closed, and moreover closed under dependent functions spaces for families of synthetic preorders.
It is also complete and cocomplete, with limits computed as in the ambient universe and colimits computed by applying the reflection operator $\RxP$ to the those of the ambient universe~\citep{rijke:2019,rijke-shulman-spitters:2020}.

\paragraph{Discrete Types}

We introduce a notion of \emph{discrete} types that provides a sufficient condition for being a synthetic preorder.

\begin{restatable}{lemma}{LemDiscreteAreSynthPreorder}\label{lem:discrete}
  A type $X$ is called \emph{discrete} when it is orthogonal to the unique map $\II \to \unit$.
  Any discrete type is a synthetic preorder.
\end{restatable}

Synthetic preorders are not closed under all dependent sum types; dependent sum types can be formed, however, when the indexing type is discrete.
\begin{restatable}{lemma}{LemDiscretelyIndexedSums}\label{lem:discretely-indexed-sums}
  Let $X$ be a discrete type, and let $x:X\vdash Yx$ be a family of synthetic preorders. Then $\sigJ{X}[x]{Y x}$ is a synthetic preorder.
\end{restatable}

\subsection{Modeling Computational Effects}\label{sec:algebra-models}

In this section, we describe how to instantiate the constructs of \decalf{} using synthetic preorders.
Assuming the extensional dependent type theory at hand has a pair of universes $\UU\in\VV$; judgments of \decalf{} are interpreted in the outer universe $\VV$. We will write $\UU_{\RxP}\subseteq\UU$ for the subuniverse of $\UU$ spanned by synthetic preorders.

\subsubsection{Value Types}
We interpret the universe of value types $\tpv$ as $\UU_\RxP$ itself, letting the decoding function $\tmv{X} \colon \tpv\to\jdg$ be the image of $X:\UU_\RxP$ under the inclusion $\UU_\RxP\hookrightarrow\VV$, which we shall leave implicit in our informal notation.

\subsubsection{Cost Structure}\label{sec:cost-structure-semantics}

The theory of \decalf{} is parameterized by a cost monoid $\costty : \tpv$ that is algorithmic, \ie becomes a singleton when $\ExtOpn$ is true.
Such a cost monoid can be defined as a quotient inductive type using the path preorder.
For instance, in \cref{sec:examples} we used the natural numbers equipped with the usual numerical ordering as the cost model $\costty$.
This synthetic preorder $\costty \eqdef \omega$ may be defined
by means of a quotient inductive type computed within $\UU_{\RxP}$, shown in \cref{fig:omega}.
The type $\omega$ is algorithmic because $\II$ is algorithmic.

\begin{figure}
  \iblock{
    \mhang{\Kwd{data}~\omega : \UU_\RxP ~\Kwd{where}}{
      \mrow{\mathsf{zero} : \omega}
      \mrow{\mathsf{suc} : \omega \to \omega}
      \mrow{\_ : \prn{n : \omega} \to n \pathrel{\omega} \mathsf{suc}\,n}
    }
  }
  \Description{Quotient inductive type defining cost structure $\omega$.}
  \caption{%
    Quotient inductive type defining cost structure $\omega$.
  }
  \label{fig:omega}
\end{figure}

\subsubsection{Computation Types}\label{sec:monads}

Let $M : \UU_\RxP \to \UU_\RxP$ be the free monad generated by some discretely-indexed finitary algebraic effect theory in $\UU_\RxP$.
We can adapt $M$ to support a cost effect using the \emph{cost monad transformer} corresponding to the writer monad $\mathbb{C}\times-$, defining a new monad $T$ on $\UU_\RxP$ by $TX \eqdef M\prn{\mathbb{C}\times X}$.
For pure code (as in \calf{}), let $M$ be the identity monad; for probabilistic choice, let $M$ be the \emph{free convex space} monad (axiomatized in \cref{fig:effect-probability}).
We interpret the universe of computation types $\tpc$ by the type of $T$-algebras.
Then $\tmc{A}$ is interpreted the same as $\tmv{\U{A}}$. Thus, we have a free-forgetful adjunction $\F\dashv \U \colon \tpc\to \tpv$ interpreting the call-by-push-value adjunctive structure of \decalf{}.

\subsubsection{Inequality Relation}
For any type $X:\tpv$ and pair of elements $x,x':\tmv{X}$, the inequality judgment $\leVal{x}{x'}{X}$ is interpreted by the path relation $x\pathrel{\tmv{X}} x'$, which is transitive because $\tpv \eqdef \UU_\RxP$.
The path relation $x\pathrel{\tmv{X}} x'$ is a proposition (and so a synthetic preorder) and may thus be internalized as a type of \decalf{}.

\subsection{Soundness and Completeness}\label{sec:semantics:sound-complete}

In summary, to interpret the inequality of \decalf{}, it suffices to work in an extensional dependent type theory with a directed interval $(\II, \irel, 0, 1)$.
Such a type theory may itself be interpreted in a \emph{QWI-topos} $\ECat$ equipped with such an interval; a QWI-topos is a cartesian closed category with finite limits, a subobject classifier, and \emph{QWI-types}, a form of quotient inductive types~\citep{fiore-pitts-steenkamp:2021}.

\begin{definition}\label{def:spt-model}
  We say that a QWI-topos $\ECat$ is a \emph{model of synthetic preorder theory} when it comes equipped with a directed interval.
\end{definition}

\subsubsection{Well-Adapted Models}

\NewDocumentCommand\Emb{}{N}

A priori, a cost bound in the synthetic theory has little to do with an \emph{actual} cost bound in concrete preorders. When a model satisfies the following \emph{well-adaptedness} property, however, we may relate a synthetic cost bound to a traditional cost bound in concrete preorders.
\begin{definition}\label{def:well-adapted-model}
  Let $\ECat$ be a model of synthetic preorder theory.
  We say that $\ECat$ is \emph{well-adapted} when there is a fully faithful cartesian closed functor $\Emb \colon \PREORD \hookrightarrow \ECat$ that preserves the interval object. This means that $\Emb \brc{0 \le 1} = \II$ and the points $0,1 : \II$ are determined by $0,1 \colon \brc{\ast} \to \brc{0 \le 1}$, respectively.
\end{definition}
Now, we discuss how to relate synthetic cost bounds to cost bounds in concrete preorders in well-adapted models.

\subsubsection{Completeness}

Let $\ECat$ be a well-adapted model of synthetic preorder theory.
The associated embedding \emph{preserves} concrete inequalities in a preorder as synthetic inequalities:

\begin{restatable}{theorem}{ThmPreorderToSyntheticPreorder}\label{thm:preorder-to-synthetic-preorder}
  Let $P : \PREORD$ be a concrete preorder. If $p \le_P p'$, then $\Emb p \pathrel{\Emb P} \Emb p'$ holds in $\ECat$.
\end{restatable}
\begin{proof}
  Unfolding the definition of the path relation, this means we need to define a path $\II \to \Emb P$ whose endpoints are determined by $\Emb p$ and $\Emb p'$ at $0$ and $1$, respectively. Using the fact that $\Emb $ is full and faithful and preserves the interval object, it suffices to define a map $\alpha : \brc{0 \le 1} \to P$, which we may define by $\alpha(0) \eqdef p$ and $\alpha(1) \eqdef p'$ since we have assumed that $p \le_P p'$. To check that the map so defined has the correct endpoints, we compute along the boundary, using the fact that $\Emb\brc{0 \le 1} \cong \II$:
  \[
  \begin{tikzpicture}[diagram]
    \node (A) {1};
    \node (E) [right = 1em of A] {$\cong$};
    \node (B) [right = 1.5em of E] {$\Emb\brc{\ast}$};
    \node (C) [right = 6em of B] {$\Emb\brc{0 \le 1}$};
    \node (D) [right = 5em of C] {$\Emb P$};
    \draw[->, bend left] (B) to node[above] {$\Emb 0$} (C);
    \draw[->, bend right] (B) to node[below] {$\Emb 1$} (C);
    \draw[->] (C) to node[above] {$\Emb \alpha$} (D);
  \end{tikzpicture}
  \]
  The top composite computes to $\Emb (\alpha(0)) = \Emb (p)$ and the bottom $\Emb (\alpha(1)) = \Emb (p')$, as desired.
\end{proof}

\begin{corollary}[Completeness of synthetic cost bounds]
  Let $P, Q : \PREORD$ be concrete preorders. Given monotone functions $f, f' \colon P \to Q$ such that $f \le f'$ on the pointwise order, then the synthetic preorder relation $\Emb f \pathrel{\Emb P \to \Emb Q} \Emb f'$ holds in $\ECat$.
\end{corollary}

\subsubsection{Soundness}

Let $\ECat$ be a well-adapted model of synthetic preorder theory.
The associated embedding \emph{reflects} synthetic inequalities as concrete inequalities in a preorder:

\begin{restatable}{theorem}{ThmSyntheticPreorderToPreorder}\label{thm:synthetic-preorder-to-preorder}
  Let $P : \PREORD$ be a concrete preorder. If $x \pathrel{\Emb P} x'$ holds in $\ECat$, then there exist $p, p' : P$ such that $p \le_P p'$ and $x = \Emb p$ and $x' = \Emb p'$.
\end{restatable}

\begin{proof}
  Unfolding the definition of the path relation, we have a path ${\alpha : \II \to \Emb P}$ from $x$ to $x'$. Because $\Emb $ is fully faithful and preserves the interval, this path determines a unique monotone map $f:\brc{0 \le 1} \to P$. Taking $p$ and $p'$ to be the endpoints of this map, we check that they are sent by the embedding to $x = \alpha(0)$ and $x' = \alpha(1)$, respectively:
  \[
  \begin{tikzpicture}[diagram]
    \node (A) {1};
    \node (E) [right = 1em of A] {$\cong$};
    \node (B) [right = 1.5em of E] {$\Emb \brc{\ast}$};
    \node (C) [right = 6em of B] {$\Emb \brc{0 \le 1}$};
    \node (D) [right = 5em of C] {$\Emb P$};
    \draw[->, bend left] (B) to node[above] {$\Emb 0$} (C);
    \draw[->, bend right] (B) to node[below] {$\Emb 1$} (C);
    \draw[->] (C) to node[above] {$\Emb f$} (D);
  \end{tikzpicture}
  \]
  The top composite computes to $\alpha(0) = \Emb (f(0)) = \Emb (p)$ and the bottom $\alpha(1) = \Emb (f(1)) = \Emb (p')$, as desired.
\end{proof}

\begin{corollary}[Soundness of synthetic cost bounds]
  Let $P, Q : \PREORD$ be concrete preorders. Given maps $x, x' \colon \Emb P \to \Emb Q$ such that the synthetic preorder relation $x \pathrel{\Emb P \to \Emb Q} x'$ holds in $\ECat$, then there exist $f, f' \colon P \to Q$ such that $f \le f'$ in the pointwise order and $x = \Emb f$ and $x' = \Emb f'$.
\end{corollary}

\subsection{A Model in Augmented Simplicial Sets}\label{sec:model-aug}

The canonical well-adapted model of synthetic preorder theory is (augmented) simplicial sets, for which the corresponding embedding is given by the nerve functor.

\subsubsection{Simplicial Sets for Synthetic Preorders}

To model synthetic preorders in a QWI-topos, we take a cue from higher category theory and consider \emph{simplicial sets}, which are presheaves on the simplex category defined below.

\begin{definition}[Simplex category]
  We will write $\SPX$ for the \emph{simplex category}, \ie the category of inhabited finite ordinals $\brk{n}$ and order-preserving maps between them. By convention, $\brk{0}$ will denote the singleton ordinal, \ie the  terminal object of $\SPX$.
\end{definition}

\paragraph{What do simplicial sets have to do with preorders?}
Every preorder can be reconstructed by gluing simplices together in a canonical
way; this is the \emph{density} of the embedding $I\colon
  \SPX\hookrightarrow\PREORD$, which implies that the corresponding nerve functor
\begin{align*}
  &\Emb \colon \PREORD \to \Psh{\SPX} \\
  &\Emb P \eqdef \Hom{\PREORD}{I - }{P}
\end{align*}
is fully faithful, \ie $\Psh{\SPX}$ is well-adapted. The synthetic
preorder theory of simplicial sets, then, studies sufficient conditions in the
internal language for arbitrary simplicial sets to ``behave like'' those that
arise from actual preorders via the nerve functor $\Emb$.

\begin{restatable}{theorem}{ThmSimplicialSetsModel}\label{thm:simplicial-sets-model}
  The category $\Psh{\SPX}$ of simplicial sets forms a model of \emph{synthetic preorder theory} in which the interval is given by the representable presheaf $\II\eqdef\Yo\brk{1}$ and its two global points.
\end{restatable}

\subsubsection{Augmented Simplicial Sets for a Nontrivial Behavioral Phase}

Unfortunately, in $\Psh{\SPX}$, the behavioral phase $\ExtOpn\eqdef 0 =_\II 1$ is the false proposition, meaning that all behavioral facts hold vacuously.
Our goal is to find a \emph{nontrivial} model $\ECat$, \ie where the slice $\ECat/\ExtOpn$ is not the terminal category. The most canonical choice for such a topos is obtained by \emph{freely extending} $\Psh{\SPX}$ with a maximal topos-theoretic point by forming an ``inverted Sierpi\'nski cone'', \ie the Artin gluing~\citep{sga:4} of the constant presheaves functor $\SET\to\Psh{\SPX}$. This gluing can also be presented equivalently by presheaves on a different category, as adding a maximal point to a presheaf topos corresponds (dually) to freely extending the base category by an initial object --- which amounts in the case of $\SPX$ to the use of \emph{augmented simplicial sets}.

\begin{definition}[Augmented simplex category]
  We will write $\AugSPX$ for the \emph{augmented simplex category}, the free extension of $\SPX$ by an initial object $\brk{-1}$. Concretely, $\AugSPX$ can be thought of as the category of \emph{arbitrary} finite ordinals and order-preserving maps between them, where $\brk{-1}$ corresponds to the empty ordinal.
\end{definition}

\begin{restatable}{theorem}{ThmAugSimplicialSetsModel}\label{thm:aug-simplicial-sets-model}
  The category $\Psh{\AugSPX}$ of \emph{augmented simplicial sets} forms a nontrivial well-adapted model of synthetic preorder theory in which
  the interval is given by the representable presheaf $\II\eqdef \Yo\brk{1}$.
  The phase distinction $\ExtOpn \eqdef 0 =_\II 1$ is isomorphic to the representable subterminal presheaf $\Yo\brk{-1}$.
\end{restatable}
Thus, for the computational effects considered in this work, we have a nontrivial model of the \decalf{} theory in $\Psh{\AugSPX}$ that is sound and complete with respect to an embedding of the category of preorders and monotone functions.
\section{Conclusion}\label{sec:conclusion}

In this work, we presented \decalf{}, an inequational extension of \calf{}~\citep{niu-sterling-grodin-harper:2022} that supports precise and approximate bounds on the cost and effect structure of programs.
Ab initio, the theory of \decalf{} has been forged and guided by the pragmatic struggles encountered in cost analysis and program verification.
Throughout the development, our guiding principle has been that a \emph{cost bound for an effectful program should be another effectful program}.

In \cref{sec:examples}, we demonstrated this methodology through a variety of case studies.
For pure, first-order algorithms, we were able to provide simple proofs of combined cost and correctness.
Such proofs in \decalf{} are more streamlined than their \calf{} counterparts and can be carried out without reference to any separable notion of recurrence.
Instead, the code itself serves the role of the recurrence, which we then solve for a closed form.
Then, using the behavioral modality, we were able to extract behavioral equalities from inequality-based program bounds.
This approach scaled naturally to support more complex classes of programs,
including higher-order programs with non-cost effects.

In \cref{sec:semantics}, we justified this style of reasoning using the notion
of a \emph{synthetic preorder theory}, a novel formulation of
an \emph{intrinsic} theory of preorders that smoothly integrates with the existing
cost--behavior phase distinction of \calf{}. To obtain a model of
this new theory, we draw inspiration from both work in directed type theory and
synthetic domain theory and characterize the synthetic preorders using simple
orthogonality conditions, which furnish a well-behaved subuniverse that
supports the structures for workaday program verification.

\begin{acks}
  We are grateful to Marcelo Fiore, Runming Li, and Parth Shastri for many insightful discussions.
  Additionally, we thank the anonymous reviewers for their thoughtful comments.

  This work was supported in part by \grantsponsor{AFOSR}{AFOSR}{https://www.afrl.af.mil/AFOSR/} (Tristan Nguyen, program manager) under grants \grantnum{AFOSR}{MURI FA9550-15-1-0053}, \grantnum{AFOSR}{FA9550-19-1-0216}, \grantnum{AFOSR}{FA9550-21-0009}, and \grantnum{AFOSR}{FA9550-23-1-0728} and in part by the \grantsponsor{NSF}{National Science Foundation}{https://nsf.gov} under award number \grantnum{NSF}{CCF-1901381}, and by \grantsponsor{AFRL}{AFRL}{https://afresearchlab.com/} through the \grantnum{AFRL}{NDSEG} fellowship.
  This work was co-funded by the European Union under the \grantsponsor{MSCA}{Marie Sk\l{}odowska-Curie Actions}{https://cordis.europa.eu/programme/id/HORIZON.1.2/en} Postdoctoral Fellowship grant agreement \grantnum[https://cordis.europa.eu/project/id/101065303]{MSCA}{101065303}.
  Any opinions, findings and conclusions or recommendations expressed in this material are those of the authors and do not necessarily reflect the views of the AFOSR, NSF, AFRL, the European Union, or the European Commission. Neither the European Union nor the granting authority can be held responsible for them.
\end{acks}

\bibliography{references/refs-bibtex,bib}

\clearpage
\appendix
\section{Value Types in \decalf{}}\label{sec:decalf-def}

\begin{figure}[ht]
  \begin{align*}
    {+} &: \tpv \to \tpv \to \tpv\\
    \mathsf{inl} &: \tmv{X} \to \tmv{\sumty{X}{Y}}\\
    \mathsf{inr} &: \tmv{Y} \to \tmv{\sumty{X}{Y}}\\
    \mathsf{case} &: \impl{\isof{P}{\tmv{\sumty{X}{Y}} \to \jdg}} (\isof{s}{\tmv{\sumty{X}{Y}}}) \\
      &\quad \to ((\isof{x}{\tmv{X}}) \to P(\inl{x})) \\
      &\quad \to ((\isof{y}{\tmv{Y}}) \to P(\inr{y})) \\
      &\quad \to P(s) \\
    \mathsf{case/}\beta_1 &:
      \sumcase{\inl{x}}{P}{e_1}{e_2} = e_1(x)\\
    \mathsf{case/}\beta_2 &:
      \sumcase{\inr{y}}{P}{e_1}{e_2} = e_1(y)\\
    \mathsf{case/}\eta &:
      \sumcase{s}{P}{\lambda x.P(\inl{x})}{\lambda y.P(\inr{y})} = P(s)
  \end{align*}
  \caption{The sum type constructor in \decalf{}.}
\end{figure}

\begin{figure}[ht]
  \begin{align*}
    \mathsf{List} &: \tpv \to \tpv\\
    \nilex &: \tmv{\listty{X}}\\
    \consex{}{} &: \tmv{X} \to \tmv{\listty{X}} \to \tmv{\listty{X}}\\
    \mathsf{rec} &: \impl{\isof{P}{\tmv{\listty{X}} \to \jdg}} (\isof{l}{\tmv{\listty{X}}}) \\
      &\quad \to {P(\nilex)} \\
      &\quad \to ((\isof{x}{\tmv{X}})~(\isof{l}{\tmv{\listty{X}}}) \to {P(l)} \to {P(\consex{x}{l})}) \\
      &\quad \to {P(l)}\\
    \mathsf{rec/}\beta_\mathsf{\nilex} &: \listrec{\nilex}{P}{e_0}{e_1} = e_0\\
    \mathsf{rec/}\beta_\mathsf{\consex{}{}} &: \listrec{\consex{x}{l}}{P}{e_0}{e_1} = e_1(x)(l)(\listrec{l}{P}{e_0}{e_1})
  \end{align*}
  \caption{The list type constructor in \decalf{}.}
\end{figure}

\section{Extended Discussion of Semantics}\label{sec:extended-semantics}

\LemDiscretelyIndexedSums*
\begin{proof}
  Boundary separation is obviously preserved by dependent sum types, so we will consider only path-transitivity.
  We fix an orthogonal lifting scenario as follows:
  \[
    \begin{tikzpicture}[diagram]
      \node (nw) {$\II\lor \II$};
      \node[right = of nw] (ne) {$\sigJ{X}[x]{Yx}$};
      \node[below = 1.5cm of nw] (sw) {$\II_2$};
      \draw[->] (nw) to node[left] {$\tau$} (sw);
      \draw[->] (nw) to node[above] {$\psi$} (ne);
    \end{tikzpicture}
  \]
  We construct a unique lift $\hat{\psi} \colon \II_2\to \sigJ{X}[x]{Yx}$ with $\psi = \hat{\psi}\circ \tau$. As $X$ is discrete, it follows that the restriction $\Prj[1]\circ \psi \colon \II\lor \II \to X$ is constant on some element $x:X$, and so the restriction $\Prj[2]\circ \psi$ is a non-dependent function $\II\lor\II\to Y x$. Therefore, it suffices to solve the following simpler orthogonal lifting problem:
  \[
    \begin{tikzpicture}[diagram]
      \node (nw) {$\II\lor \II$};
      \node[right = of nw] (ne) {$Y x$};
      \node[below = 1.5cm of nw] (sw) {$\II_2$};
      \draw[->] (nw) to node[left] {$\tau$} (sw);
      \draw[->] (nw) to node[above] {$\psi\circ\Prj[2]$} (ne);
      \draw[->,exists] (sw) to node[sloped,below] {$\exists!$} (ne);
    \end{tikzpicture}
  \]
  The unique lift above exists because we have assumed each $Y x$ is path-transitive.
\end{proof}

\LemDiscreteAreSynthPreorder*
\begin{proof}
  This follows from \cref{lem:discretely-indexed-sums}, considering the dependent sum type $\sigJ{X}[\_]{\unit}$.
\end{proof}

\end{document}